\documentclass[12pt]{article}

\usepackage[latin1]{inputenc}
\usepackage[american]{babel}
\usepackage{enumerate}
\usepackage{amsfonts,amsmath,amssymb,amsthm,amstext,latexsym,paralist}
\usepackage{url}
\usepackage{xspace}
\usepackage[ruled,lined,boxed,commentsnumbered]{algorithm2e}
\usepackage{tikz}
\usepackage{subfig}
\usepackage{verbatim}

\usepackage{times}

\newtheorem{lemma}{Lemma}
\newtheorem{theorem}{Theorem}
\newtheorem{proposition}{Proposition}
\newtheorem{corollary}{Corollary}

\newcommand{\continuous}{\textsc{Continuous}\xspace}
\newcommand{\discrete}{\textsc{Discrete}\xspace}
\newcommand{\VDD}{\textsc{Vdd-Hopping}\xspace}
\newcommand{\incremental}{\textsc{Incremental}\xspace}
\newcommand{\smin}{s_{\mathit{min}}} 
\newcommand{\smax}{s_{\mathit{max}}} 
\newcommand\II{\ensuremath{\mathcal{I}}\xspace}
\newcommand\MinE{\ensuremath{\textsc{Min\-En\-er\-gy}(G,D)}\xspace}

\begin{document}

\title{Reclaiming the energy of a schedule: \\ models and algorithms
}

\author{Guillaume Aupy, Anne Benoit,  Fanny Dufoss\'e and Yves Robert\\
LIP, \'Ecole Normale Supérieure de Lyon, France\\[.1cm]
\{Guillaume.Aupy$|$Anne.Benoit$|$Fanny.Dufosse$|$Yves.Robert\}@ens-lyon.fr
}

\date{April 2012}

\maketitle

\footnotetext{
A two-page extended abstract of this work appeared as a short presentation in SPAA'2011, 
while the long version has been accepted for publication in ``Concurrency and Computation: Practice and Experience''.}

\begin{abstract}
We consider a task graph to be executed on a set of processors. We
assume that the mapping is given, say by an ordered list of tasks to
execute on each processor, and we aim at optimizing the energy
consumption while enforcing a prescribed bound on the execution time.
While it is not possible to change the allocation of a task, it is
possible to change its speed. Rather than using a local approach such
as backfilling, we consider the problem as a whole and
study the impact of several speed variation models on its
complexity. For continuous speeds, we give a closed-form formula for
trees and series-parallel graphs, and we cast the problem
into a geometric programming problem for general directed acyclic
graphs. We show that the classical dynamic voltage and frequency
scaling (DVFS) model with discrete modes leads to a NP-complete
problem, even if the modes are regularly distributed (an important
particular case in practice, which we analyze as the incremental
model). On the contrary, the VDD-hopping model leads to a polynomial
solution. Finally, we provide an approximation algorithm for the
incremental model, which we extend for the general DVFS model.
\end{abstract}

\section{Introduction}
\label{sec.intro}


The \emph{energy consumption} of computational platforms has recently
become a critical problem, both for economic and environmental
reasons~\cite{greenMills}.  As an example, the Earth Simulator
requires about 12~MW (Mega Watts) of peak power, and PetaFlop systems
may require 100~MW of power, nearly the output of a small power plant
(300~MW).  At \$100 per MW.Hour, peak operation of a PetaFlop machine
may thus cost \$10,000 per hour~\cite{1105799}.  Current estimates
state that cooling costs \$1 to \$3 per watt of heat
dissipated~\cite{980157}. This is just one of the many economical
reasons why energy-aware scheduling has proved to be an important
issue in the past decade, even without considering battery-powered
systems such as laptops and embedded systems.  As an example, the Green500 list (\url{www.green500.org})
provides rankings of the most energy-efficient supercomputers in the
world, therefore raising even more awareness about power consumption.

To help reduce energy dissipation, processors can run at different
speeds.  Their power consumption is the sum of a static part (the cost
for a processor to be turned on) and a dynamic part, which is a
strictly convex function of the processor speed, so that the execution
of a given amount of work costs more power if a processor runs in a
higher mode~\cite{10.1109/IPDPS.2006.1639597}. More precisely, a
processor running at speed $s$ dissipates $s^3$
watts~\cite{280894,pruhsTCS,pow3,pow3IPDPS,pow3ICPP} per time-unit,
hence consumes $s^3 \times d$ joules when operated during $d$ units of
time.  Faster speeds allow for a faster execution, but they also lead
to a much higher (supra-linear) power consumption.

Energy-aware scheduling aims at minimizing the energy
consumed during the execution of the target application. Obviously,
it makes sense only if it is coupled with some performance bound
to achieve, otherwise, the optimal solution always is to run each processor at
the slowest possible speed.

In this paper, we investigate energy-aware scheduling strategies for
executing a task graph on a set of processors. The main originality is
that we assume that the mapping of the task graph is given, say by an
ordered list of tasks to execute on each processor. There are
many situations in which this problem is important, such as optimizing
for legacy applications, or accounting for
affinities between tasks and resources, or even when tasks are pre-allocated~\cite{Rayward95}, for example for security
reasons.
In such situations, assume that a list-schedule has been computed for
the task graph, and that its execution time should not exceed a
deadline~$D$.
We do not have the freedom to change the assignment of a given task,
but we can change its speed to reduce energy consumption, provided
that the deadline~$D$ is not exceeded after the speed change.
Rather than using a local approach such as
backfilling~\cite{Wang2010,Prathipati2004}, which only reclaims gaps
in the schedule, we consider the problem as a whole, and we assess the
impact of several speed variation models on its complexity.
More precisely, we investigate the following models:
\begin{description} 
\item[\continuous model.] Processors can have arbitrary speeds, and
  can vary them continuously:
this model is unrealistic (any possible value of the speed, say
$\sqrt{e^{^\pi}}$, cannot be obtained) but it is theoretically
appealing~\cite{BKP07}. A maximum speed,~$\smax$, cannot be exceeded.
\item[\discrete model.] Processors have a discrete number
of predefined speeds (or frequencies), which correspond to different
voltages that the processor can be subjected
to~\cite{Okuma2001}.
Switching frequencies is not allowed during the execution of a given
task, but two different tasks scheduled on a same processor can be
executed at different frequencies.
\item[\VDD model.] This model is similar to the \discrete one, except
  that switching modes during the execution of a given task is
  allowed: any rational speed can be simulated,
by simply switching, at the appropriate time during the execution of a
task, between two consecutive modes~\cite{Miermont2007}.
\item[\incremental model.]  In this variant of the \discrete model, we
  introduce a value $\delta$ that corresponds the minimum permissible
  speed increment,
induced by the minimum voltage increment that
can be achieved when controlling the processor CPU. This new model
aims at capturing a realistic version of the \discrete model, where
the different modes are spread regularly instead of arbitrarily
chosen.
\end{description} 

Our main contributions are the following.
For the \continuous model, we give a closed-form formula for trees and
series-parallel graphs, and we cast the problem
into a geometric programming problem~\cite{Boyd2004} for
general DAGs.
For the \VDD model, we show that the optimal solution for general DAGs
can be computed in polynomial time, using a (rational) linear
program. Finally,
for the \discrete and \incremental models, we show that the problem is
NP-complete.
Furthermore, we provide approximation algorithms which rely on the
polynomial algorithm for the \VDD model,
and we compare their solution with the optimal
\continuous solution.

The paper is organized as follows.
We start with a survey of related literature in
Section~\ref{sec.related}. We then provide the formal description of
the framework and of the energy models in Section~\ref{sec.framework},
together with a simple example to illustrate the different models.
The next two sections constitute the heart of the paper: in
Section~\ref{sec.continuous}, we provide analytical formulas for
continuous speeds, and the formulation into the convex optimization
problem. In Section~\ref{sec.alldiscrete}, we assess the
complexity of the problem with all the discrete models: \discrete,
\VDD and \incremental, and we discuss approximation algorithms.
Finally we conclude in Section~\ref{sec.conclusion}.

\section{Related work}
\label{sec.related}

\smallskip
Reducing the energy consumption of computational platforms is an
important research topic, and many techniques at the process, circuit
design, and micro-architectural levels have been
proposed~\cite{Lee2000,Lahiri2002,Grosse2009}.
The dynamic voltage and frequency scaling (DVFS) technique has been
extensively studied, since it may lead to efficient energy/performance
trade-offs~\cite{JPG04,1105799,BKP07,ChenKuo07,Buyya2007,Yang2009,Wang2010}.
Current microprocessors (for instance, from AMD~\cite{amd} and
Intel~\cite{intel}) allow the speed to be set
dynamically. Indeed, by lowering supply voltage, hence processor
clock frequency, it is possible to achieve important reductions in
power consumption, without necessarily increasing the execution time.
We first discuss different optimization problems that arise in this
context. Then we review energy models.

\subsection{DVFS and optimization problems}

When dealing with energy consumption,
the most usual optimization function consists in
minimizing the energy consumption, while ensuring a deadline on the
execution time (i.e., a real-time constraint), as discussed in the
following papers.

In~\cite{Okuma2001}, Okuma et al. demonstrate that voltage scaling is
far more effective than the shutdown approach, which simply stops the
power supply when the system is inactive. Their target processor
employs just a few discretely variable voltages.
De Langen and Juurlink~\cite{juurLan09} discuss leakage-aware
scheduling heuristics which investigate both DVS and processor
shutdown, since static power consumption due to leakage current is
expected to increase significantly.
Chen et al.~\cite{Raghavan2005} consider parallel sparse applications,
and they show that when scheduling applications modeled by
a directed acyclic graph with a well-identified critical path, it is
possible to lower the voltage during non-critical execution of
tasks, with no impact on the execution time.
Similarly, Wang et al.~\cite{Wang2010} study the slack time for
non-critical jobs, they extend their execution time and thus reduce
the energy consumption without increasing the total execution time.
Kim et al.~\cite{Buyya2007} provide power-aware scheduling algorithms
for bag-of-tasks applications with deadline constraints, based on
dynamic voltage scaling. Their goal is to minimize power consumption as
well as to meet the deadlines specified by application users.

For real-time embedded systems, slack reclamation techniques are used.
Lee and Sakurai~\cite{Lee2000} show how to exploit slack time arising
from workload variation, thanks to a software feedback control of
supply voltage. Prathipati~\cite{Prathipati2004} discusses
techniques to take advantage of run-time variations in the execution
time of tasks; it determines the minimum voltage under which each task
can be executed, while guaranteeing the deadlines of each task. Then,
experiments are conducted on the Intel StrongArm SA-1100 processor, which has
eleven different frequencies, and the Intel PXA250 XScale embedded
processor with four frequencies.
In~\cite{XuBig}, the goal of Xu et al. is to schedule a set of
independent tasks,
given a worst case execution cycle (WCEC) for each task, and a global
deadline, while accounting for time and energy penalties when the
processor frequency is changing. The frequency of the processor can be
lowered when some slack is obtained dynamically, typically when a task
runs faster than its WCEC.
Yang and Lin~\cite{Yang2009} discuss algorithms with preemption,
using DVS techniques; substantial energy can be saved using these
algorithms, which succeed to claim the static and dynamic slack time,
with little overhead.

Since an increasing number of systems are powered by batteries,
maximizing battery life also is an important optimization problem.
Battery-efficient systems can be obtained with similar techniques of
dynamic voltage and frequency scaling, as described by Lahiri et
al. in~\cite{Lahiri2002}.
Another optimization criterion is the energy-delay product,
since it accounts for a trade-off between performance and energy
consumption, as for instance discussed by Gonzalez and Horowitz
in~\cite{Gonzalez1996}.
We do not discuss further these latter optimization problems, since our
goal is to minimize the energy consumption, with a fixed deadline.

In this paper, the application is a task graph
(directed acyclic graph), and we assume that the mapping, i.e., an
ordered list of tasks to execute on each processor, is given. Hence,
our problem is closely related to slack reclamation techniques, but
instead on focusing on non-critical tasks as for instance
in~\cite{Wang2010}, we consider the problem as a whole.
Our contribution is to perform an exhaustive complexity study for
different energy models. In the next paragraph, we discuss related work
on each energy model.


\subsection{Energy models}

Several energy models are considered in the literature, and they can
all be categorized in one of the four models investigated in this
paper, i.e.,  \continuous, \discrete, \VDD or \incremental.

The \continuous model is used mainly for theoretical studies. For
instance,  Yao et al.~\cite{YDS95}, followed by Bansal et
al.~\cite{BKP07}, aim at scheduling a collection
of tasks (with release time, deadline and amount of work), and the
solution is the time at which each task is scheduled, but also, the
speed at which the task is executed. In these papers, the speed can
take any value, hence following the \continuous model.

We believe that the most widely used model is the \discrete
one. Indeed, processors
have currently only a few discrete number of possible
frequencies~\cite{amd,intel,Okuma2001,Prathipati2004}. Therefore, most
of the papers discussed above follow this model. Some studies
exploit the continuous model to determine the smallest frequency
required to run a task, and then choose the closest upper discrete
value, as for instance~\cite{Prathipati2004} and~\cite{ZHC02}.

Recently, a new local dynamic voltage scaling architecture has been
developed, based on
the \VDD model~\cite{Miermont2007,Beigne2008b,Beigne2008}.
It was shown in~\cite{Lee2000} that significant power can be saved by
using two distinct voltages, and architectures using this principle
have been developed (see for instance~\cite{Kawaguchi2001}).
Compared to traditional power converters, a new design with no needs
for large passives or costly technological options has been validated
in a STMicroelectronics CMOS 65nm
low-power technology~\cite{Miermont2007}.

To the best of our knowledge, this paper introduces the \incremental
model for the first time. The main
rationale 
is that future technologies may well have an increased number of
possible frequencies, and these will follow a regular pattern. For
instance, note that the SA-1100 processor, considered
in~\cite{Prathipati2004}, has eleven frequencies which are
equidistant, i.e., they follow the \incremental model.  Lee and
Sakurai~\cite{Lee2000} exploit discrete levels of clock frequency as
$f$, $f/2$, $f/3$, ..., where $f$ is the master (i.e., the higher)
system clock frequency. This model is closer to the \discrete model,
although it exhibits a regular pattern similarly to the \incremental
model.

Our work is the first attempt to compare these different models: on
the one hand, we assess the impact of the model on the problem
complexity (polynomial vs NP-hard), and on the other hand, we provide
approximation algorithms building upon these results.
The closest work to ours is the paper by Zhang et al.~\cite{ZHC02}, in
which the authors also consider the mapping of directed acyclic
graphs, and compare the \discrete and the \continuous models.
We go beyond their work in this paper, with an exhaustive complexity
study, closed-form formulas for the continuous model,
and the comparison with the \VDD and \incremental models.


\section{Framework}
\label{sec.framework}

First we detail the optimization problem in Section~\ref{sec.pb}.
Then we describe the four energy models in Section~\ref{sec.mod}.
Finally, we illustrate the models and motivate the problem with an
example in Section~\ref{sec.ex}.


\subsection{Optimization problem}
\label{sec.pb}

Consider an application task graph $\mathcal{G}=(V,\mathcal{E})$,
with $n=|V|$ tasks
denoted as $V= \{T_1, T_2, \dots, T_n\}$,
and where the set $\mathcal{E}$ denotes the precedence edges between
tasks. Task~$T_i$ has a cost~$w_i$ for $1 \leq i \leq n$.
We assume that the tasks in $\mathcal{G}$ have been allocated onto a parallel
platform made up of identical processors.
We define the \emph{execution graph} generated by this allocation as
the graph $G=(V,E)$, with the following augmented set of edges:
\begin{compactitem}
\item $\mathcal{E} \subseteq E$: if an edge exists in the precedence graph, it
  also exists in the execution graph;
\item if $T_1$ and $T_2$ are executed successively, in this order, on the same
processor, then $(T_1,T_2)\in~\!E$.
\end{compactitem}

The goal is to the minimize the energy consumed during the execution
while enforcing a deadline~$D$ on the execution time.
We formalize the optimization problem in the simpler case where each
task is executed at constant speed. This strategy is optimal for the
\continuous model (by a convexity argument) and for the \discrete and
\incremental models (by definition).
For the \VDD model, 
we reformulate the problem in Section~\ref{sec.vdd}.
Let $d_i$ be the duration of the execution of task $T_i$, $t_i$ its
completion time, and $s_i$ the speed at which it is executed.
We obtain the following formulation of the \MinE problem,
given an execution graph $G= (V,E)$ and a deadline~$D$;
the $s_i$~values are variables, whose values are constrained by the
energy model (see Section~\ref{sec.mod}).
\begin{equation}
\label{eq.opt}
\begin{array}{lrl}
\text{~~~Minimize}& & \sum_{i=1}^{n}  s_i^3 \times d_i\\
\text{~~~subject to} &
\text{(i)} & w_i = s_i \times d_i \text{ for each task  } T_i \in V\\
&\text{(ii)} & t_i + d_j \leq t_j \text{ for each edge  } (T_i,T_j) \in E\\
&\text{(iii)} & t_i \leq D \text{ for each task  } T_i \in V 
\end{array}
\end{equation}

Constraint (i) states that the whole task can be executed in time
$d_i$ using speed $s_i$.
Constraint (ii) accounts for all dependencies, and constraint (iii)
ensures that the execution time does not exceed the deadline~$D$.
The energy consumed throughout the execution is the objective
function. It is the sum, for each task, of the energy consumed by this
task, as we detail in the next section.
Note that $d_i=w_i/s_i$, and
therefore the objective function can also be expressed as
$\sum_{i=1}^{n}  s_i^2 \times w_i$.

\subsection{Energy models}
\label{sec.mod}


In all models, when a processor operates at speed~$s$ during
$d$~time-units, the corresponding consumed energy
is $s^3 \times d$, which is the dynamic part of the energy
consumption, following the classical models of the
literature~\cite{280894,pruhsTCS,pow3,pow3IPDPS,pow3ICPP}.
Note that we do not take static energy into account, because
all processors are up and alive during the whole execution.
We now detail the possible speed values in each energy model,
which should be added as a constraint in Equation~(\ref{eq.opt}).
\begin{compactitem}
\item
In the \continuous model,
processors can have arbitrary speeds, from $0$ to a maximum
value~$\smax$, and a processor can change its speed at any time during
execution.
\item
In the \discrete model, processors have a set of possible
speed values, or modes, denoted as  $s_1,...,s_m$. There is no
assumption on the range and distribution of these modes.
The speed of a processor cannot change during
the computation of a task, but it can change from task to task.
\item
In the \VDD model,
a processor can run at different speeds $s_1,...,s_m$,
as in the previous model, but it can also change its speed during a
computation. The energy consumed during the execution of one task is
the sum, on each time interval with constant speed~$s$, of the energy
consumed during this interval at speed~$s$.
\item
In the \incremental model, we introduce a value~$\delta$ that
 corresponds to the minimum permissible speed (i.e., voltage)
 increment. That means that possible speed values are obtained as
$s=\smin + i\times \delta$, where $i$~is an integer such that
$0\leq i \leq \frac{\smax-\smin}{\delta}$.
Admissible speeds lie in the interval $[\smin,\smax]$.
This new model aims at capturing a realistic version of the \discrete
model, where the different modes are spread regularly between
$s_1=\smin$ and $s_m=\smax$, instead of being arbitrarily chosen.
It is intended as the modern counterpart of a potentiometer knob!
\end{compactitem}

\vspace{-.4cm}
\subsection{Example}
\label{sec.ex}
\vspace{-.4cm}

Consider an application with four tasks of costs $w_1=3$, $w_2=2$,
$w_3=1$ and $w_4=2$, and one precedence constraint
$T_1\rightarrow T_3$. We assume that $T_1$ and $T_2$ are allocated, in
this order, onto processor $P_1$, while $T_3$ and $T_4$
are allocated, in this order, on processor~$P_2$.
The resulting execution graph~$G$ is given in
Figure~\ref{example}, with two precedence constraints
added to the initial task graph.
The deadline on the execution time is~$D=1.5$.

We set the maximum speed to~$\smax=6$ for the
\continuous model. For the \discrete and \VDD
models, we use the set of speeds $s^{(d)}_1=2$, $s^{(d)}_2=5$ and
$s^{(d)}_3=6$. Finally, for the \incremental model, we set $\delta=2$,
$\smin=2$ and $\smax=6$,
so that possible speeds are $s^{(i)}_1=2$, $s^{(i)}_2=4$ and $s^{(i)}_3=6$.
We aim at finding the optimal execution speed~$s_i$ for each
task~$T_i$ ($1\leq i \leq 4$), i.e., the values of~$s_i$ which
minimize the energy consumption.

With the \continuous model, the optimal speeds are non rational 
values, and we obtain
$$s_1 = \frac23(3+35^{1/3})\simeq 4.18;
\quad s_2 = s_1\times \frac2{35^{1/3}}\simeq 2.56;
\quad s_3 = s_4 = s_1\times \frac3{35^{1/3}}\; \simeq 3.83.$$

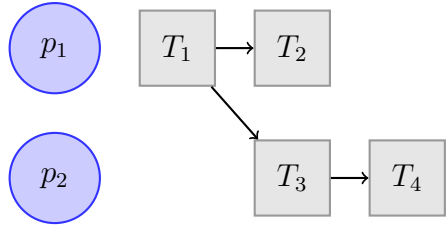
\begin{figure}
\begin{center}
\begin{tikzpicture}[scale=0.25]
\tikzstyle{proc}=[circle,
                                    thick,
                                    minimum size=1.2cm,
                                    draw=blue!80,
                                    fill=blue!20]
\tikzstyle{task}=[rectangle,
                                    thick,
                                    minimum size=1cm,
                                    draw=gray!80,
                                    fill=gray!20]
\matrix[row sep=0.5cm,column sep=0.5cm] {
        \node (p0) [proc] {$p_{1}$}; &
        \node (v0) [task] {$T_1$};       &
        \node (v1) [task] {$T_2$}; &
	&
        \\
        \node (p1) [proc] {$p_2$}; &
	&
        \node (v2)   [task] {$T_3$};       &
        \node (v3) [task] {$T_4$}; &
        \\
    };
    \path[->]
        (v0) edge[thick] (v1)
        (v0) edge[thick] (v2)
        (v2) edge[thick] (v3)
        ;
\end{tikzpicture}
\end{center}
 \vspace{-.5cm}
\caption{Execution graph for the example.}
\label{example}
\end{figure}

Note that all speeds are lower than the maximum~$\smax$.
These values are obtained thanks to the formulas derived in
Section~\ref{sec.continuous}.
The energy consumption is then
$E^{(c)}_{opt} = \sum_{i=1}^4 w_i \times s_i^2
  = 3.s_1^2+2.s_2^2+ 3.s_3^2 \simeq 109.6$.
The execution time is
$\frac{w_1}{s_1} + \max\left(\frac{w_2}{s_2},
  \frac{w_3+w_4}{s_3}\right)$,
and with this solution, it is equal to the deadline~$D$ (actually,
both processors reach the deadline, otherwise we could slow down the
execution of one task). 

For the \discrete model, if we execute all tasks at speed
$s^{(d)}_2=5$, we obtain an energy $E=8\times 5^2=200$. A better
solution is obtained with $s_1=s^{(d)}_3=6$, $s_2=s_3=s^{(d)}_1=2$ and
$s_4=s^{(d)}_2=5$, which turns out to be optimal:
$E^{(d)}_{opt}=3\times36+(2+1)\times4+2\times25=170$.
Note that $E^{(d)}_{opt}>E^{(c)}_{opt}$, i.e., the optimal energy
consumption with the \discrete model is much higher than the one
achieved with the \continuous model. Indeed, in this case, even though
the first processor executes during $3/6 + 2/2 = D$ time units, the second
processor remains idle since $3/6 + 1/2 + 2/5=1.4<D$. The problem
turns out to be NP-hard (see Section~\ref{sec.np}),
and the solution has been found by performing an
exhaustive search.

With the \VDD model, 
we set $s_1=s^{(d)}_2=5$; for the other tasks, we run part of the time at
speed~$s^{(d)}_2=5$, and part of the time at speed~$s^{(d)}_1=2$ in
order to use the idle time and lower the energy consumption.
$T_2$~is executed at speed $s^{(d)}_1$ during time~$\frac{5}{6}$
and at speed $s^{(d)}_2$ during time~$\frac{2}{30}$ (i.e., the
first processor executes during time $3/5 + 5/6 + 2/30 = 1.5 = D$,
and all the work for~$T_2$ is done: $2\times 5/6 + 5\times 2/30 = 2 = w_2$).
$T_3$~is executed at speed~$s^{(d)}_2$ (during time $1/5$),
and finally $T_4$~is executed at speed $s^{(d)}_1$
during time~$0.5$ 
and at speed $s^{(d)}_2$ during time~$1/5$ 
(i.e., the second processor executes during time
$3/5 + 1/5 + 0.5 + 1/5 = 1.5 = D$,
and all the work for~$T_4$ is done: $2\times 0.5 + 5\times 1/5 = 2 = w_4$).
This set of speeds turns out to be optimal (i.e., it is the optimal
solution of the linear program introduced in Section~\ref{sec.vdd}),
with an energy consumption
$E^{(v)}_{opt}=(3/5 + 2/30 + 1/5 + 1/5) \times 5^3 + (5/6 + 0.5)
\times 2^3 = 144$. As expected, $E^{(c)}_{opt}\leq E^{(v)}_{opt} \leq
E^{(d)}_{opt}$, i.e., the \VDD solution stands between the optimal
\continuous solution, and the more constrained \discrete solution.

For the \incremental model, the reasoning is similar to the \discrete
case, and the optimal solution is obtained by an exhaustive search:
all tasks should be executed at speed~$s^{(i)}_2=4$, with an energy
consumption~$E^{(i)}_{opt}=8 \times 4^2 = 128 > E^{(c)}_{opt}$.
It turns out to be
better than \discrete and \VDD, since it has different discrete values
of energy which are more appropriate for this example.

\section{The \continuous model}
\label{sec.continuous}

With the \continuous model, processor speeds can take any value
between $0$ and~$\smax$. First we prove that, with this model, the
processors do not change their speed during the execution of a task
(Section~\ref{sec.prel}).  Then, we derive in Section~\ref{sec.spec}
the optimal speed values
for special execution graph structures, expressed as closed form
algebraic formulas, and we show that these values may be irrational
(as already illustrated in the example in Section~\ref{sec.ex}).
Finally, we formulate the problem for general DAGs as a convex
optimization program 
in Section~\ref{sec.gen}.

\subsection{Preliminary lemma}
\label{sec.prel}

\begin{lemma}[constant speed per task]
\label{lem.prel}
With the \continuous model, each task is executed at constant speed,
i.e., a processor does not change its speed during the execution of
a task.
\end{lemma}

\begin{proof}
Suppose that in the optimal solution, there is a task whose
speed changes during the execution. Consider the first
time-step at which the change occurs:
the computation begins at speed~$s$ from time~$t$ to time~$t'$, and
then continues at speed~$s'$ until time~$t''$.
The total energy consumption for this task in the time
interval~$[t;t'']$ is $E=(t'-t)\times s^3+(t''-t')\times (s')^3$.
Moreover, the amount of work done for this task is
$W=(t'-t)\times s+(t''-t')\times s'$.

If we run the task during the whole interval~$[t;t'']$ at constant
speed~$W/(t''-t)$, the same amount of work is done within the same time.
However, the energy consumption during this interval of time is now
 $E'=(t''-t)\times (W/(t''-t))^3$.
By convexity of the function $x \mapsto x^3$, we obtain $E'<E$
since $t<t'<t''$. This contradicts the hypothesis of optimality of the
first solution, which concludes the proof.
\end{proof}

\subsection{Special execution graphs}
\label{sec.spec}

\subsubsection{Independent tasks}

Consider the problem of minimizing the energy of $n$ independent
tasks (i.e., each task is mapped onto a distinct processor, and there
are no precedence constraints in the execution graph),
while enforcing a deadline~$D$.
\begin{proposition}[independent tasks]
\label{ntask}
When $G$ is composed of independent tasks $\{T_1, \dots, T_n\}$,
the optimal solution to \MinE is obtained when each
task~$T_i$ ($1\leq i \leq n$) is computed at speed
$s_i=\frac{w_i}{D}$. If there is a task~$T_i$ such
that $s_i > \smax$, then the problem has no solution.
\end{proposition}
\begin{proof}
 For task~$T_i$, the speed~$s_i$ corresponds to the slowest speed at
 which the processor can execute the task,
so that the deadline is not exceeded.
If $s_i  > \smax$, the corresponding processor will never be able to
complete its execution before the deadline, therefore there is no
solution.
To conclude the proof, we note that any other solution would have
higher values of~$s_i$ because of the deadline constraint, and hence a
higher energy consumption. Therefore, this solution is optimal.
\end{proof}

\subsubsection{Linear chain of tasks}

This case corresponds for instance to $n$~independent tasks $\{T_1,
\dots, T_n\}$ executed
onto a single processor. The execution graph is then a linear chain
(order of execution of the tasks), with $T_i \rightarrow T_{i+1}$,
for $1\leq i < n$.




\begin{proposition}[linear chain]
	\label{ntask1proc}
$\;\;$ When $G$ is a linear chain of tasks, the optimal solution to \MinE is
obtained when each task is executed at
speed $s=\frac{W}{D}$, with $W=\sum_{i=1}^n w_i$.\\
If $s>\smax$, then there is no solution.
\end{proposition}

\begin{proof}
%
Suppose that in the optimal solution, tasks~$T_i$ and~$T_j$ are such
that $s_i < s_j$. 
The total energy consumption is~$E_{opt}$.
We define $s$ such that the execution of both tasks running at
speed~$s$ takes the same amount of time than in the optimal solution,
i.e., $(w_i+w_j)/s = w_i/s_i + w_j/s_j$:
$s = \frac{(w_i+w_j)}{w_is_j + w_js_i} \times s_is_j$.
Note that $s_i < s < s_j$ (it is the barycenter of two points with
positive mass).

We consider a solution such that the speed
of task~$T_k$, for $1\leq k \leq n$, with $k\neq i$ and $k \neq j$, is
the same as in the optimal solution, and the speed of tasks $T_i$
and~$T_j$ is~$s$. By definition of~$s$, the execution time has not
been modified. The energy consumption of this solution is~$E$,
where $E_{opt}-E = w_is_i^{2} + w_js_j^{2} -(w_i+w_j)s^{2}$, i.e.,
the difference of energy with the optimal solution is only impacted by
tasks $T_i$ and~$T_j$, for which the speed has been modified.
By convexity of the function $x \mapsto x^{2}$, we obtain
$E_{opt}>E$, which contradicts its optimality.
Therefore, in the optimal solution, all tasks have the same execution
speed. Moreover, the energy consumption is minimized when the speed is
as low as possible, while the deadline is not exceeded. Therefore, the
execution speed of all tasks is $s=W/D$.
\end{proof}

\begin{corollary}
\label{cor.chain}
A linear chain with $n$~tasks is equivalent to a single task of cost
$W=\sum_{i=1}^n w_i$.
\end{corollary}

\noindent Indeed, in the optimal solution, the $n$ tasks are executed at the
same speed, and they can be replaced by a single task of cost~$W$,
which is executed at the same speed and consumes  the
same amount of energy.

\subsubsection{Fork and join graphs}

Let $V\!=\!\{T_1,\dots,T_n\}$.
We consider either a fork graph $G = (V\cup\{T_{0}\}, E)$,
with $E=\{(T_{0},T_i),T_i \in V\}$, or a join graph $G =
(V\cup\{T_{0}\}, E)$, with $E=\{(T_i,T_{0}),T_i\in V\}$.
$T_{0}$~is either the source of the fork or the sink of
the join.

\begin{theorem}[fork and join graphs]
\label{opt_fork}
When $G$ is a fork (resp. join) execution graph with $n+1$ tasks
$T_0,T_1,\dots,T_n$, the optimal solution to \MinE is the following:\\
${ }\quad \bullet$ the execution speed of the source (resp. sink)~$T_0$ is
$s_{0} = \dfrac{\left(\sum_{i=1}^n w_i^{3}\right)^{\frac{1}{3}} +
  w_{0}}{D}$\; ; \\
${ }\quad \bullet$ for the other tasks~$T_i$, $1\leq i \leq n$, we have
$s_i = s_{0} \times \dfrac{w_i}{\left(\sum_{i=1}^n
    w_i^{3}\right)^{\frac{1}{3}}}$
if $s_{0}\leq\smax$\; .\\[.2cm]
Otherwise, $T_{0}$ should be executed at speed
$s_0=\smax$, and the other speeds are $s_i = \frac{w_i}{D'}$, with
$D'=D-\frac{w_{0}}{\smax}$, if they do not exceed~$\smax$
(Proposition~\ref{ntask} for independent tasks).
Otherwise there is no solution.

If no speed exceeds $\smax$, the corresponding
energy consumption is
\[
{\bf minE}(G,D)=\frac{\left((\sum_{i=1}^n w_i^3)^{\frac{1}{3}} +
    w_0\right)^3}{D^2}\; .
\]
\end{theorem}

\begin{proof}
Let $t_0=\frac{w_{0}}{s_0}$. Then, the source or the sink requires a
time~$t_0$ for execution. For $1\leq i \leq n$, task~$T_i$
must be executed within a time~$D-t_0$ so that the deadline is
respected. Given~$t_0$, we can compute the speed~$s_i$ for task~$T_i$
using Theorem~\ref{ntask}, since the tasks are independent:
$s_i = \frac{w_i}{D-t_0} = w_i \cdot \frac{s_0}{s_0 D - w_0}$.
The objective is therefore to minimize
$\sum_{i=0}^n w_i s_i^{2}$, which is a function of~$s_0$:

$$\sum_{i=0}^n w_i s_i^{2} = w_0s_0^2 +
\sum_{i=1}^n w_i^3 \cdot \frac{s_0^2}{(s_0D - w_0)^2}
= s_0^2 \left( w_0 +  \frac{\sum_{i=1}^n w_i^3}{(s_0 D - w_0)^2}\right)
= f(s_0).
$$

\noindent Let $W_3 = \sum_{i=1}^n w_i^3$.
In order to find the value of $s_0$ which minimizes this function, we
study the function~$f(x)$, for $x>0$.
$f'(x) = 2x \left(w_{0} + \frac{W_3}{(x D-w_{0})^2}\right)  -
2D \cdot x^{2} \cdot \frac{W_3}{(x D-w_{0})^{3}} $, and therefore
$f'(x)=0$ for $x=(W_3^{\frac{1}{3}} + w_0)/D$.
We conclude that the optimal speed for task~$T_0$ is
$s_0 = \frac{\left(\sum_{i=1}^n w_i^{3} \right)^{\frac{1}{3}} +
  w_{0}}{D}$, if $s_0 \leq \smax$. Otherwise, $T_0$~should be executed
at the maximum speed $s_0 =\smax$, since it is the bottleneck task.
In any case, for $1\leq i \leq n$, the optimal speed for task~$T_i$ is
$s_i = w_i \frac{s_0}{s_0 D - w_0}$.

Finally, we compute the exact expression of {\bf minE}$(G,D) =
f(s_0)$, when $s_0\leq \smax$:
$$f(s_0) = s_0^2 \left( w_0 +  \frac{W_3}{(s_0 D - w_0)^2}\right)
	=\left(\frac{W_3^{\frac{1}{3}} +
    w_0}{D}\right)^2 \left(\frac{W_3}{W_3^{2/3}}
 +w_0 \right )
=\frac{\left(W_3^{\frac{1}{3}} +
    w_0\right)^3}{D^2}, $$ 
which concludes the proof.
\end{proof}

\begin{corollary}[equivalent tasks for speed]
	\label{cor.fork.speed}
  Consider a fork or join graph with tasks~$T_i$, $0\leq i \leq n$,
  and a deadline~$D$, and assume that the speeds in the optimal
  solution to \MinE do not
  exceed~$\smax$.
   Then, these speeds are the same as in the optimal solution for  $n+1$
  independent tasks $T'_0, T'_1, \dots, T'_n$, where
  $w'_0 = \left (\sum_{i=1}^n w_{i}^{3} \right )^{\frac{1}{3}} + w_{0}$,
  and, for $1\leq i \leq n$,
  $w'_i=w'_0 \cdot\frac{w_i}{\left (\sum_{i=1}^n w_{i}^{3}
  \right )^{\frac{1}{3}}}\;$.
\end{corollary}

\begin{corollary}[equivalent task for energy]
	\label{cor.fork.energy}
Consider a fork or join graph $G$ and a deadline~$D$, and assume that the speeds
in the optimal solution to \MinE do not exceed~$\smax$.
We say that the graph~$G$ is \emph{equivalent} to the graph
$G^{(eq)}$, consisting of a single task $T^{(eq)}_0$
of weight
$w^{(eq)}_0 = \left (\sum_{i=1}^n w_{i}^{3} \right )^{\frac{1}{3}} + w_{0}$,
because the minimum energy consumption of both graphs are identical:
{\bf minE}$(G,D)$={\bf minE}$(G^{(eq)},D)$.
\end{corollary}

\subsubsection{Trees}

We extend the results on a fork graph for a tree
$G=(V,E)$ with $|V|=n+1$ tasks. Let $T_0$ be the root of the tree;
it has $k$~children tasks, which are each themselves the root of a
tree. A tree can therefore be seen as a fork graph, where the tasks of
the fork are trees.

The previous results for fork graphs naturally lead to an algorithm
that peels off branches of the tree, starting with the leaves, and
replaces each fork subgraph in the tree, composed of a root $T_0$ and
$k$ children, by one task (as in Corollary~\ref{cor.fork.energy})
which
becomes the unique child of $T_0$'s parent in the tree. We say that
this task is {\em equivalent} to the fork graph, since the optimal
energy consumption will be the same.  The computation of the {\em
  equivalent} cost of this task is done thanks to a call to the {\bf
  eq} procedure, while the {\bf tree} procedure computes the
solution to \MinE (see
Algorithm~\ref{algo.tree}). 
Note that the algorithm computes the minimum energy for a tree, but it
does not return the speeds at which each task must be
executed. However, the algorithm returns the speed of the root task,
and it is then straightforward to compute the speed of each children
of the root task, and so on.

\begin{theorem}[tree graphs]
\label{th.tree}
When $G$ is a tree rooted in~$T_0$ ($T_0 \in V$, where $V$ is the set
of tasks), 
 the optimal solution to \MinE can be computed in polynomial time
 $O(|V|^2)$.
\end{theorem}

\begin{proof}
Let $G$ be a tree graph rooted in~$T_0$.
 The optimal solution to \MinE is obtained with a call to {\bf tree}
  $(G,T_0,D)$, and we prove its optimality recursively on the depth of
  the tree.  Similarly to the case of the fork graphs, we reduce the
  tree to an equivalent task which, if executed alone within a
  deadline~$D$, consumes exactly the same amount of energy. The
  procedure {\bf eq} is the procedure which reduces a tree to its
  equivalent task  (see Algorithm~\ref{algo.tree}).

  If the tree has depth~$0$, then it is a single task, {\bf
    eq}~$(G,T_0)$ returns the equivalent cost~$w_0$, and the optimal
   execution speed is $\frac{w_0}{D}$ (see
  Proposition~\ref{ntask}). There is a solution if and only if this
  speed is not greater than~$\smax$, and then the corresponding energy
  consumption is $\frac{w_0^3}{D^2}$, as returned by the algorithm.

Assume now that for any tree of depth~$i<p$, {\bf eq} computes its
equivalent cost, and {\bf tree} returns its optimal energy
consumption. We consider a tree~$G$ of depth~$p$ rooted in~$T_0$:
$G = T_0 \cup \{G_i\}$, where each subgraph~$G_i$ is a tree, rooted
in~$T_i$, of maximum depth~$p-1$.
As in the case of forks, we know that each subtree~$G_i$ has a
deadline $D-x$, where $x=\frac{w_0}{s_0}$, and $s_0$ is the speed at
which task~$T_0$ is executed.
By induction hypothesis, we suppose that each graph $G_i$ 
is equivalent to a single task, $T'_i$, of cost $w'_i$
(as computed by the procedure {\bf eq}).
We can then use the results obtained on forks to compute $w^{(eq)}_0$
(see proof of Theorem~\ref{opt_fork}):\\[-.6cm]
\begin{center}$w^{(eq)}_0 = \displaystyle \left ( \sum_{i} (w'_{i})^{3} \right
)^{\frac{1}{3}} + w_{0}$.\end{center}

Finally the tree is equivalent to one task of cost $w^{(eq)}_0$, and
if $\frac{w^{(eq)}_0}{D}\leq \smax$, the energy consumption is
 $\frac{\left(w^{(eq)}_0\right)^3}{D^2}$, and no speed
 exceeds~$\smax$.

Note that the speed of a task is always greater than the speed of its
successors. Therefore, if $\frac{w^{(eq)}_0}{D}> \smax$, we execute
the root of the tree at speed~$\smax$ and then process each
subtree~$G_i$ independently. Of course, there is no solution if
$\frac{w_0}{\smax}>D$, and otherwise we perform the recursive calls to
{\bf tree} to process each subtree independently. Their deadline is
then $D-\frac{w_0}{\smax}$.

\medskip
To study the time complexity of this algorithm, first note
that when calling {\bf tree}~$(G,T_0,D)$, there might be at most $|V|$
recursive calls to {\bf tree}, once at each node of the tree. Without
accounting for the recursive calls, the {\bf tree} procedure performs
one call to the {\bf eq} procedure, which computes the cost of the
equivalent task. This {\bf eq} procedure takes a time~$O(|V|)$, since
we have to consider the $|V|$ tasks, and we add the costs one by
one. Therefore, the overall complexity is in~$O(|V|^2)$.
\end{proof}

\begin{algorithm}
\caption{Solution to \MinE for trees.\label{algo.tree}
}

procedure {\bf tree} (tree $G$, root $T_0$, deadline~$D$)\\
\Begin{
Let $w$={\bf eq} (tree $G$, root $T_0$)\;
\eIf{$\frac{w}{D}\leq \smax$ }
  {\Return $\frac{w^3}{D^2}$;}
{\eIf{$\frac{w_0}{\smax} > D$}
  {\Return{Error:No Solution;}}
{ {\em /* $T_0$ is executed at speed $\smax$ */ }\\
  \Return{$w_0 \times \smax^2 + \displaystyle \!\!\!\sum_{G_i\mbox{ subtree
        rooted in }T_i\in\mbox{children}(T_0)} \!\!\!{\bf tree} \left(G_i,T_i,D-\frac{w_0}{\smax}\right)$}\;
    }}
}

\medskip

procedure {\bf eq} (tree $G$, root $T_0$)\\
\Begin{
\eIf{children($T_0$)=$\emptyset$}
  {\Return $w_0$;}
  {
  \Return{$\left(\displaystyle \sum_{G_i\mbox{ subtree rooted in }T_i\in\mbox{children}(T_0)}
      \left(\mbox{\bf eq} (G_i,T_i)\right)^3 \right)^{\frac{1}{3}} +
    w_0$
};
  }
}
\end{algorithm}

\subsubsection{Series-parallel graphs}

We can further generalize our
results to series-parallel graphs (SPGs), which are built from a
sequence of compositions (parallel or series) of smaller-size SPGs.
The smallest SPG consists of two nodes connected by an edge (such a
graph is called an {\em elementary SPG}). The first
node is the source, while the second one is the sink of the SPG. When
composing two SGPs in series, we merge the sink of the first SPG with
the source of the second one. For a parallel composition, the two
sources are merged, as well as the two sinks, as illustrated
in Figure~\ref{Compo}.

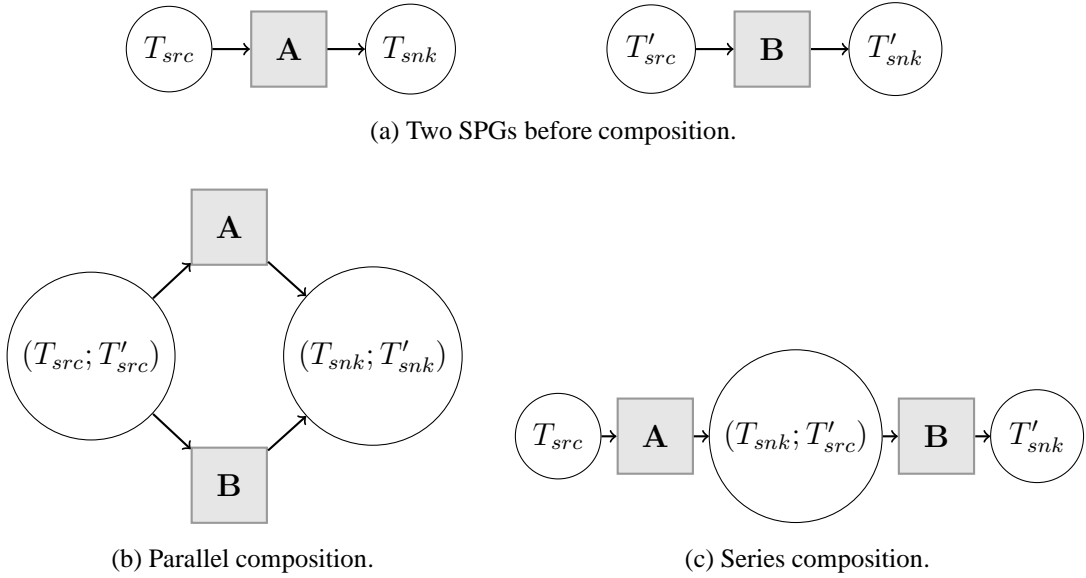
\begin{figure}
\begin{center}
 \subfloat[Two SPGs before composition.] 
{
\begin{tikzpicture}[scale=0.1]
\tikzstyle{state}=[circle,
                                     draw]

\tikzstyle{tree}=[rectangle,
                                    thick,
                                    minimum size=1cm,
                                    draw=gray!80,
                                    fill=gray!20]
  \matrix[row sep=0.1cm,column sep=0.5cm,ampersand replacement=\&] {
        \node (v0) [state] {$T_{src}$}; \&
        \node (A) [tree] {$\mathbf{A}$};\&
        \node (vn) [state] {$T_{snk}$}; \&\&\&\&

        \node (v'0) [state] {$T'_{src}$}; \&
        \node (B)   [tree] {$\mathbf{B}$};\&
        \node (v'n) [state] {$T'_{snk}$}; \&
        \\
    };
    \path[->]
        (v0) edge[thick] (A)
        (A) edge[thick] (vn)
        (v'0) edge[thick] (B)
        (B)   edge[thick] (v'n)
        ;
 \end{tikzpicture}
\label{SP_before}
}

 \subfloat[Parallel composition.]
{
\begin{tikzpicture}[scale=0.1]
\tikzstyle{state}=[circle,
                                    draw]
 \tikzstyle{tree}=[rectangle,
                                    thick,
                                    minimum size=1cm,
                                    draw=gray!80,
                                    fill=gray!20]
  \matrix[row sep=0.01cm,column sep=0.2cm,ampersand replacement=\&] {
        \&
        \node (A) [tree] {$\mathbf{A}$}; \&
        \&
        \\
        \node (v0) [state] {$(T_{src};T'_{src})$}; \&
        \&
        \node (vn)   [state] {$(T_{snk};T'_{snk})$};     \&
        \\
        \&
        \node (B)   [tree] {$\mathbf{B}$};       \&
        \&
        \\
	};
    \path[->]
        (v0) edge[thick] (A)
        (A) edge[thick] (vn)
        (v0) edge[thick] (B)
        (B)   edge[thick] (vn)
        ;
 \end{tikzpicture}
\label{Par-Compo}
}
\subfloat[Series composition.]
{
\begin{tikzpicture}[scale=0.1]
\tikzstyle{state}=[circle,
                                    draw]
 \tikzstyle{tree}=[rectangle,
                                    thick,
                                    minimum size=1cm,
                                    draw=gray!80,
                                    fill=gray!20]
  \matrix[row sep=0.01cm,column sep=0.2cm,ampersand replacement=\&] {
	\&
	\&
	$~$\&
	\&
	\&
	\\
        \node (v0) [state] {$T_{src}$}; \&
        \node (A) [tree] {$\mathbf{A}$}; \&
        \node (v'0)   [state] {$(T_{snk};T'_{src})$};     \&
        \node (B)   [tree] {$\mathbf{B}$};       \&
        \node (v'n) [state] {$T'_{snk}$}; \&
        \\
	\&
	\&
	$~$\&
	\&
	\&
	\\
    };
    \path[->]
        (v0) edge[thick] (A)
        (A) edge[thick] (v'0)
        (v'0) edge[thick] (B)
        (B)   edge[thick] (v'n)
        ;

\end{tikzpicture}
\label{Ser-Compo}
}
\caption{Composition of series-parallel graphs (SPGs).}
\label{Compo}
\end{center}
\vspace{-.5cm}
\end{figure}

We can extend the results for tree graphs to SPGs, by replacing step
by step the SPGs by an equivalent task (procedure {\bf
  cost} in Algorithm~\ref{algo.spg}): we can compute the equivalent
cost for a series or parallel composition.

However, since it is no longer true that the speed of a task is always
larger than the speed of its successor (as was the case in a tree),
we have not been able to find a recursive property on the tasks that
should be set to~$\smax$, when one of the speeds obtained with the
previous method exceeds~$\smax$. The problem of computing a closed
form for a SPG with a finite value of~$\smax$ remains open. Still,
we have the following result when $\smax = +\infty$:

\begin{theorem}[series-parallel graphs]
\label{th.spg}
When $G$ is a SPG, it is possible to compute recursively a closed form
expression of the optimal solution of \MinE, assuming $\smax =
+\infty$, in polynomial time~$O(|V|)$, where $V$~is the set of tasks.
\end{theorem}

\begin{proof}
  Let $G$ be a series-parallel graph. The optimal solution to \MinE is
  obtained with a call to {\bf SPG}~$(G,D)$, and we prove its
  optimality recursively. Similarly to trees, the main idea is to
  peel the graph off, and to transform it until there remains only a
  single equivalent task which, if executed alone within a
  deadline~$D$, would consume exactly the same amount of energy. The
  procedure {\bf cost} is the procedure which reduces a tree to its
  equivalent task (see Algorithm~\ref{algo.spg}).

  The proof is done by induction on the number of compositions
  required to build the graph~$G$,~$p$. If $p=0$, $G$~is an elementary
  SPG consisting in two tasks, the source~$T_0$ and the sink~$T_1$. It
  is therefore a linear chain, and therefore equivalent to a single
  task whose cost is the sum of both costs, $w_0+w_1$ (see
  Corollary~\ref{cor.chain} for linear chains). The procedure {\bf
    cost} returns therefore the correct equivalent cost, and {\bf SPG}
  returns the minimum energy consumption.

  Let us assume that the procedures return the correct equivalent cost
  and minimum energy consumption for any SPG consisting of $i<p$
  compositions. We consider a SPG~$G$, with $p$ compositions. By
  definition, $G$ is a composition of two smaller-size SPGs, $G_1$
  and~$G_2$, and both of these SPGs have strictly fewer than
  $p$~compositions. We consider $G'_1$ and~$G'_2$, which are identical
  to $G_1$ and $G_2$, except that the cost of their source and sink
  tasks are set to~$0$ (these costs are handled separately), and we
  can reduce both of these SPGs to an equivalent task, of respective
  costs $w'_1$ and~$w'_2$, by induction hypothesis. There are two
  cases:

\begin{compactitem}
 \item  If $G$ is a series composition, then after the reduction of
   $G'_1$ and~$G'_2$, we have a linear chain in which we consider the
   source~$T_0$ of~$G_1$, the sink~$T_1$ of~$G_1$ (which is also the
   source of~$G_2$), and the sink~$T_2$ of~$G_2$. The equivalent
   cost is therefore $w_0 + w'_1 + w_1 + w'_2 + w_2$, thanks to
   Corollary~\ref{cor.chain} for linear chains.

\item If $G$ is a parallel composition,
the resulting graph is a fork-join graph, and we can use Corollaries
\ref{cor.chain} and~\ref{cor.fork.energy} to compute the cost of the
equivalent task, accounting for the source~$T_0$  and the sink~$T_1$:
$w_0 + \left((w'_1)^3+(w'_2)^3\right)^\frac{1}{3} + w_1$.
\end{compactitem}

Once the cost of the equivalent task of the SPG has been computed with
the call to {\bf cost}~$(G)$, the optimal energy consumption is
$\frac{\left({\bf cost}(G)\right)^3}{D^2}$.

\medskip
Contrarily to the case of tree graphs, since we never need to call the
{\bf SPG} procedure again because there is no constraint on~$\smax$,
the time complexity of the algorithm is the complexity of the {\bf
  cost} procedure. There is exactly one call to {\bf cost} for each
composition, and the number of compositions in the SPG is
in~$O(|V|)$. All operations in {\bf cost} can be done in~$O(1)$,
hence a complexity in $O(|V|)$.
\end{proof}

\begin{algorithm}
\caption{Solution to \MinE for series-parallel graphs.
\label{algo.spg}}
procedure {\bf SPG} (series-parallel graph $G$, deadline $D$)\\
\Begin{
\Return{$\frac{\left(\mbox{\bf cost}(G)\right)^3}{D^2}$};
}

procedure {\bf cost} (series-parallel graph $G$)\\
\Begin{
Let $T_0$ be the source of $G$ and $T_1$ its sink\;
\eIf{$G$ is composed of only two tasks, $T_0$ and $T_1$}
{\Return{$w_0+w_1$};}
{{\em /* $G$ is a composition of two SPGs $G_1$ and~$G_2$. */}\\
  For $i=1,2$, let $G'_i=G_i$ where the cost of source and sink tasks is
set to~$0$\; 
$w'_1=\mbox{\bf cost}(G'_1); \; w'_2=\mbox{\bf cost}(G'_2)$\;
\eIf{$G$ is a series composition}
{Let $T_0$ be the source of~$G_1$, $T_1$ be its sink, and $T_2$ be the
sink of~$G_2$\;
\Return{$w_0 + w'_1 + w_1 + w'_2 + w_2$};}
{{\em /* It is a parallel composition. */}\\
Let $T_0$ be the source of~$G$, and $T_1$ be its sink\;
\Return{$w_0 + \left((w'_1)^3+(w'_2)^3\right)^\frac{1}{3} + w_1$};}
}}
\end{algorithm}

\subsection{General DAGs}
\label{sec.gen}

For arbitrary execution graphs, we can rewrite the \MinE problem as follows:

\begin{equation}
\label{eq.opt2}
\begin{array}{lrl}
\text{~~~Minimize}& & \sum_{i=1}^{n}  u_i^{-2} \times w_i\\
\text{~~~subject to}
&\text{(i)} & t_i + w_j \times u_j \leq t_j \text{ for each edge  }
(T_i,T_j) \in E\\
&\text{(ii)}  & t_i \leq D \text{ for each task  } T_i \in V\\
&\text{(iii)} & u_i \geq \frac{1}{\smax} \text{ for each task  } T_i \in V\\
\end{array}
\end{equation}

Here, $u_i = 1/s_i$ is the inverse of the speed to execute task
$T_i$. We now have
a convex optimization problem to solve, with linear constraints in the
non-negative variables $u_i$ and $t_i$.
In fact, the objective function is a posynomial, so we have a
geometric programming problem (see~\cite[Section~4.5]{Boyd2004})
for which efficient numerical schemes exist. However, as illustrated
on simple fork graphs, the optimal speeds
are not expected to be rational numbers but instead arbitrarily
complex expressions (we have the cubic root of the sum of cubes
for forks, and nested expressions of this form for trees). From a
computational complexity point of view,
we do not know how to encode such numbers in polynomial size of the
input (the rational task weights and the execution deadline).
Still, we can always solve the problem numerically and get fixed-size
numbers which are good approximations of the optimal values.

\medskip In the following, we show that the total power consumption of
any optimal schedule is constant throughout execution.  While this
important property does not help to design an optimal solution, it
shows that a schedule with large variations in its power consumption
is likely to waste a lot of energy.

We need a few notations before stating the result. Consider a schedule
for a graph~$G=(V,E)$ with $n$ tasks.  Task $T_i$ is executed at
constant speed $s_i$ (see Lemma~\ref{lem.prel}) and during interval
$[b_i,c_i]$: $T_i$ begins its execution at time $b_i$ and completes it
at time $c_i$. The total power consumption $P(t)$ of the schedule at
time $t$ is defined as the sum of the power consumed by all tasks
executing at time $t$:
$$P(t) = \sum_{1 \leq i \leq n, \; t \in [b_i,c_i]} s_i^3\; . $$

\begin{theorem}
	\label{powconst}
Consider an instance of \continuous, and an optimal schedule for this instance, such that no speed is equal
to~$\smax$. Then the total power consumption of the schedule throughout execution is constant.
\end{theorem}

\begin{proof}
We prove this theorem by induction on the number of tasks of the
graph. First we prove a preliminary result:

\begin{lemma}
  \label{conn2}
  Consider a graph~$G=(V,E)$ with $n \geq 2$ tasks, and any optimal
  schedule of deadline~$D$.  Let $t_1$ be the earliest completion
  time of a task in the
  schedule. Similarly, let $t_2$ be the latest starting time of a task in the
  schedule. Then, either $G$ is composed of independent tasks, or
  $0<t_1\leq t_2<D$.
\end{lemma}

\begin{proof} 
Task $T_i$ is executed at speed $s_i$ and during interval $[b_i,c_i]$. We have $t_1 = \min_{1\leq i \leq n} c_i$
and $t_2 = \max_{1\leq i \leq n} b_i$. Clearly, $0 \leq t_1,t_2 \leq D$ by definition of the schedule.
Suppose that $t_2 < t_1$. Let $T_1$ be a task that
ends at time~$t_1$, and $T_2$ one that starts at time~$t_2$. Then:
\begin{compactitem}
\item $\nexists T\in V,~(T_1,T) \in E$ (otherwise, $T$ would start
  after $t_2$), therefore, $t_1=D$; 
\item $\nexists T\in V,~(T,T_2) \in E$ (otherwise, $T$ would finish
  before $t_1$); therefore $t_2=0$. 
\end{compactitem}
This also means that all tasks start at time~$0$ and end at time~$D$. Therefore,
$G$ is only composed of independent tasks.
\end{proof}

Back to the proof of the theorem, we consider first the case of a
graph with only one task. In an optimal schedule, the task is executed
in time~$D$, and at constant speed (Lemma~\ref{lem.prel}), hence with constant
power consumption.

Suppose now that the property is true for all DAGs with at most $n-1$
tasks. Let $G$ be a DAG with $n$ tasks.  If $G$ is exactly composed of $n$
independent tasks, then we know that the
power consumption of~$G$ is constant (because all task speeds are
constant). Otherwise, let $t_1$ be the earliest completion time,
and $t_2$ the latest starting time of a task in the optimal
schedule.  Thanks to Lemma~\ref{conn2}, we have  $0<t_1\leq t_2<D$.

Suppose first that $t_1=t_2=t_0$. There are three kinds of
tasks: those beginning at time~$0$ and ending at time~$t_0$ (set~$S_1$), those
beginning at time~$t_0$ and ending at time~$D$ (set~$S_2$), and finally those
beginning at time~$0$ and ending at time~$D$ (set~$S_3$). Tasks in $S_3$
execute during the whole schedule duration, at constant speed, hence their contribution to the
total power consumption $P(t)$  is the same at each time-step $t$. Therefore, we can suppress them
from the schedule without loss of generality. Next we determine the value of $t_0$. 
Let $A_1=\sum_{T_i\in S_1} w_i^3$, and
$A_2=\sum_{T_i\in S_2} w_i^3$. The energy consumption between $0$ and
$t_0$ is $\frac{A_1}{t_0^2}$, and between $t_0$ and $D$, it is
$\frac{A_2}{(D-t_0)^2}$. The optimal energy consumption is obtained with
$t_0=\frac{A_1^{\frac{1}{3}}}{A_1^{\frac{1}{3}}+A_2^{\frac{1}{3}}}$.
Then, the total power consumption of the optimal schedule is the same in both intervals, hence at each time-step: we derive that
$P(t) = \left(\frac{A_1^{\frac{1}{3}}+A_2^{\frac{1}{3}}}{D}\right)^3$, which
is constant.

Suppose now that $t_1<t_2$.  For each task~$T_i$, let $w'_i$ be the
number of operations executed before~$t_1$, and $w''_i$ the number of
operations executed after~$t_1$ (with $w'_i+w''_i=w_i$).  Let $G'$
be the DAG~$G$ 
with execution costs~$w'_i$, and $G''$ be the DAG~$G$ 
with execution costs~$w''_i$. The tasks with a cost equal to~$0$ are
removed from the DAGs. Then, both $G'$ and $G''$ have strictly fewer
than $n$~tasks. We can therefore apply the induction hypothesis. We derive that
the power consumption in both DAGs is constant.  
Since we did not change the speeds of the tasks, the total power consumption $P(t)$ in $G$
is the same as in $G'$ if $t < t_1$, hence a constant. Similarly,  
the total power consumption $P(t)$ in $G$
is the same as in $G''$ if $t > t_1$, hence a constant.
Considering the same partitioning with $t_2$ instead of $t_1$, we show that the total power consumption
$P(t)$ is a constant before $t_2$, and also a constant after $t_2$. But $t_1 < t_2$, and the intervals
$[0,t_2]$ and $[t_1,D]$ overlap. Altogether, the total power consumption is the same constant throughout $[0,D]$,
which concludes the proof.
\end{proof}

\section{Discrete models}
\label{sec.alldiscrete}
\vspace{.5cm}

In this section, we present complexity results on the three
energy models with a finite number of possible speeds.
The only polynomial instance is for the \VDD model, for which we write
a linear program in Section~\ref{sec.vdd}. Then, we give
NP-completeness results in Section~\ref{sec.np}, and approximation
results in Section~\ref{sec.approx}, for the \discrete and
\incremental models.


\subsection{The \VDD model}
\label{sec.vdd}

\begin{theorem}
\label{th.vdd}
With the \VDD model, \MinE
can be solved in polynomial time.
\end{theorem}

\begin{proof}
Let $G$ be the execution graph of an application with $n$ tasks, and $D$ a deadline.
Let $s_1,...,s_m$ be the set of possible processor speeds.
We use the
following rational variables: for $1\leq i \leq n$ and $1\leq j \leq m$,
$b_i$ is the starting time of the execution of task~$T_i$,
and $\alpha_{(i,j)}$ is the time spent at speed~$s_j$ for executing
task~$T_i$. There are $n + n\times m = n(m+1)$ such variables.
Note that the total execution time of task~$T_i$ is
$\sum_{j=1}^m \alpha_{(i,j)}$.
The constraints are: 
\begin{compactitem}
\item $\forall 1\leq i\leq n, \; b_i\geq 0$: starting times of all
  tasks are non-negative numbers;
\item $\forall 1\leq i \leq n, \; b_i+\sum_{j=1}^m
  \alpha_{(i,j)}\leq D$: the deadline is not exceeded by any task;
\item $\forall 1 \leq i,i' \leq n$ such that $T_i\rightarrow T_{i'}$,
$\; t_i+\sum_{j=1}^m \alpha_{(i,j)}\leq t_{i'}$: a task
cannot start before its predecessor has completed its execution;
\item $\forall 1\leq i \leq n, \; \sum_{j=1}^m
  \alpha_{(i,j)} \times s_j\geq w_i$:  task $T_i$ is completely executed.
\end{compactitem}

The objective function is then
  $\min\left(\sum_{i=1}^n\sum_{j=1}^m \alpha_{(i,j)}s_j^3\right)$.

The size of this linear program is clearly polynomial in the size of
the instance, all $n(m+1)$ variables are rational, and therefore it
can be solved in polynomial time~\cite{SchrijverComb}.
\end{proof}

\subsection{NP-completeness results}
\label{sec.np}

\begin{theorem}
\label{th.np}
With the \incremental model (and hence the \discrete model), \\
\MinE is NP-complete.
\end{theorem}

\begin{proof}
We consider the associated decision problem:
given an execution graph, a deadline, and a bound on the energy
consumption,
can we find an execution speed for each task such that the deadline
and the bound on energy are respected?
The problem is clearly in NP:
given the execution speed of each task, computing the execution time
and the energy consumption can be done in polynomial time.

To establish the completeness, we use a reduction from
2-Partition~\cite{GareyJohnson}. We consider an instance $\II_1$ of
2-Partition: given $n$ strictly positive integers
$a_1, \ldots, a_n$, does there exist a subset $I$ of
$\{1, \ldots, n\}$ such that
$\sum_{i\in I}a_i=\sum_{i \notin I}a_i$? Let
$T=\frac{1}{2}\sum_{i=1}^n a_i$.

\medskip
We build the following
instance~$\II_2$ of our problem: the execution graph is a linear chain
with $n$~tasks, where:
\begin{compactitem}
\item task $T_i$ has size $w_i=a_i$;
\item the processor can run at $m=2$ different speeds;
\item $s_1=1$ and $s_2=2$, (i.e., $\smin=1, \smax=2, \delta=1$);
\item $L=3T/2$;
\item $E=5T$.
\end{compactitem}
Clearly, the size of $\II_2$ is polynomial in the size of $\II_1$.

\medskip
\emph{Suppose first that instance $\II_1$ has a solution~$I$.} For all $i\in I$,
 $T_i$ is executed at speed~$1$, otherwise it is executed at speed~$2$.
The execution time is then $\sum_{i\in I}a_i+\sum_{i\notin
  I}a_i/2=\frac{3}{2}T = D$,
and the energy consumption is $E=\sum_{i\in I} a_i+\sum_{i\notin I}a_i
\times 2^2 = 5T = E$. Both bounds are respected, and therefore the
execution speeds are a solution to~$\II_2$.

\medskip
\emph{Suppose now that $\II_2$ has a solution.} Since we consider the
\discrete and \incremental models, each task run either at speed~$1$, or at
speed~$2$.
Let $I=\{i\; |\; T_i\mbox{ is executed at speed }1\}$.
Note that we have $\sum_{i\notin I}a_i = 2T - \sum_{i\in I}a_i$.

The execution time is 
$D'=\sum_{i\in I}a_i+\sum_{i\notin  I}a_i/2 = T + (\sum_{i\in I}a_i)/2$.
Since the deadline is not exceeded, $D' \leq D = 3T/2$, and therefore
$\sum_{i\in I}a_i \leq T$.

For the energy consumption of the solution
of~$\II_2$, we have $E' = \sum_{i\in I} a_i+\sum_{i\notin I}a_i
\times 2^2 = 2T + 3\sum_{i\notin I}a_i$. Since $E'\leq E=5T$, we
obtain $3\sum_{i\notin I}a_i\leq 3T$, and hence $\sum_{i\notin
  I}a_i\leq T$.
\smallskip

Since $\sum_{i\in I}a_i + \sum_{i\notin I}a_i = 2T$, we conclude that
$\sum_{i\in I}a_i = \sum_{i\notin I}a_i =T$, and therefore $\II_1$~has
a solution. This concludes the proof.
\end{proof}

\subsection{Approximation results}
\label{sec.approx}

Here we explain, for the \incremental and \discrete models,
how the solution to the NP-hard problem can be approximated.
Note that, given an execution graph and a deadline,
the optimal energy consumption with the \continuous model is
always lower than that with the other models, which are more
constrained.



\begin{theorem}
\label{approx_inc}
With the \incremental model, for any integer $K>0$, the \MinE problem
can be approximated within a factor
$(1+\frac{\delta}{\smin})^2(1+\frac{1}{K})^2$, in a time polynomial
in the size of the instance and in~$K$.
\end{theorem}

\begin{proof}
Consider an instance $\II_{inc}$ of the problem with the \incremental
model. The execution graph~$G$ has $n$~tasks, $D$~is the deadline,
$\delta$~is the minimum permissible speed increment,
and $\smin, \smax$ are the speed bounds.
Moreover, let $K>0$ be an integer, and
let $E_{inc}$~be the optimal value of the energy consumption
for this instance~$\II_{inc}$.

We construct the following instance $\II_{vdd}$ with the \VDD model:
the execution graph and the deadline are the same as in
instance~$\II_{inc}$, and the speeds can take the values
$$\left\{\smin\times\left(1+\frac{1}{K}\right)^i\right\}_{0\leq i\leq N}\;,$$
where $N$ is such that $s_{max}$ is not exceeded:
$N=\left\lfloor(\ln(\smax)-\ln(\smin))/\ln\left(1+\frac{1}{K}\right)\right\rfloor$.
As $N$ is asymptotically of order $O(K\ln(\smax))$, the number of
possible speeds in $\II_{vdd}$, and hence the size of
$\II_{vdd}$, is polynomial in the size of $\II_{inc}$ and~$K$.

\medskip
Next, we solve $\II_{vdd}$ in polynomial time thanks to
Theorem~\ref{th.vdd}. For each task~$T_i$, let $s^{(vdd)}_i$ be the
average speed of~$T_i$ in this solution: if the execution time of the
task in the solution is~$d_i$, then $s^{(vdd)}_i = w_i/d_i$;
$E_{vdd}$~is the optimal energy consumption obtained with these speeds.
Let $s^{(algo)}_i=\min_u\{\smin+u\times \delta \; | \; u\times \delta \geq
s^{(vdd)}_i\}$ be the smallest speed in~$\II_{inc}$ which is larger
than~$s^{(vdd)}_i$. There exists such a speed since, because of the
values chosen for~$\II_{vdd}$,  $s^{(vdd)}_i\leq \smax$.
The values $s^{(algo)}_i$ can be computed in time polynomial in
the size of $\II_{inc}$ and~$K$. Let $E_{algo}$ be the energy
consumption obtained with these values.

\medskip
In order to prove that this algorithm is an approximation of the
optimal solution, we need to prove that
$E_{algo} \leq (1+\frac{\delta}{\smin})^2(1+\frac{1}{K})^2\times E_{inc}$.
For each task~$T_i$, $s^{(algo)}_i-\delta\leq s^{(vdd)}_i\leq s^{(algo)}_i$.
Since $\smin \leq s^{(vdd)}_i$, we derive that $s^{(algo)}_i \leq
s^{(vdd)}_i \times (1+\frac{\delta}{\smin})$.
Summing over all
tasks, we get\\
$\text{~~~~~~~~~~~~~~} E_{algo}=\sum_iw_i\left(s^{(algo)}_i\right)^2
  \leq \sum_iw_i\left(s^{(vdd)}_i
    \times(1+\frac{\delta}{\smin})\right)^2
  \leq E_{vdd} \times \left(1+\frac{\delta}{\smin}\right)^2$.\\[.2cm]
Next, we bound $E_{vdd}$ thanks to the optimal solution with the
\continuous model,~$E_{con}$.
Let~$\II_{con}$ be the instance where
the execution graph $G$, the deadline $D$, the speeds $\smin$ and $\smax$
are the same as in instance~$\II_{inc}$, but now admissible speeds take
any value between $\smin$ and~$\smax$.
Let $s^{(con)}_i$ be the optimal continuous speed for task~$T_i$,
and let $0\leq u\leq N$ be the value such that:~\\[.2cm]
$ \text{~~~~~~~~~~~~~~} \smin\times\left(1+\frac{1}{K}\right)^u \leq s^{(con)}_i
  \leq \smin\times\left(1+\frac{1}{K}\right)^{u+1}=s^*_i~$.\\[.2cm]
In order to bound the energy consumption for~$I_{vdd}$, we assume that
$T_i$ runs at speed $s^*_i$, instead of $s^{(vdd)}_i$. The solution
with these speeds is a solution to~$I_{vdd}$, and its energy
consumption is $E^* \geq E_{vdd}$.
From the previous inequalities, we deduce that
$s^*_i \leq s^{(con)}_i \times \left(1+\frac{1}{K}\right)$, and by
summing over all tasks, \\
$E_{vdd} \leq E^* = \sum_i w_i \left(s^*_i\right)^2
  \leq \sum_i w_i \left(s^{(con)}_i \times \left(1+\frac{1}{K}\right)\right)^2
  \leq E_{con} \times \left(1+\frac{1}{K}\right)^2
 \leq E_{inc} \times \left(1+\frac{1}{K}\right)^2\; .
$
%
%
%
\end{proof}

\begin{proposition}${ }$
\begin{compactitem}
\item For any integer $\delta>0$, any instance of \MinE with the
  \continuous model
can be approximated within a factor
$(1+\frac{\delta}{\smin})^2$ in the \incremental model with speed
increment $\delta$.
\item For any integer $K>0$, any instance of \MinE with the \discrete model
can be approximated within a factor
$(1+\frac{\alpha}{s_1})^2(1+\frac{1}{K})^2$, with
$\alpha=\max_{1 \leq i < m}\{s_{i+1}-s_i\}$, in a time polynomial
in the size of the instance and in~$K$.
\end{compactitem}
\end{proposition}

\begin{proof}
For the first part, let $s^{(con)}_i$ be the optimal continuous speed
for task~$T_i$ in instance $\II_{con}$;
$E_{con}$ is the optimal energy consumption.
For any task $T_i$, let $s_i$ be the speed of $\II_{inc}$ such that
$s_i-\delta<s^{con}_i\leq s_i$.
Then, $s^{(con)}_i\leq s_i\times\left(1+\frac{\delta}{\smin}\right)$.
Let $E$ be the energy with speeds $s_i$.
$E_{con}\leq E\times \left(1+\frac{\delta}{\smin}\right)^2$.
Let $E_{inc}$ be the optimal energy of $\II_{inc}$. Then,
$E_{con}\leq E_{inc}\times
\left(1+\frac{\delta}{\smin}\right)^2$.
\smallskip

For the second part, we use the same algorithm as in Theorem \ref{approx_inc}.
The same proof leads to the approximation ratio
with $\alpha$ instead of $\delta$.
\end{proof}

\section{Conclusion}
\label{sec.conclusion}

In this paper, we have assessed the tractability of a classical
scheduling problem, with task preallocation, under various energy
models. We have given several results related to \continuous
speeds. However, while these are of conceptual importance, they cannot
be achieved with physical devices, and we have analyzed several models
enforcing a bounded number of achievable speeds, a.k.a. modes.  In the
classical \discrete model that arises from DVFS techniques, admissible
speeds can be irregularly distributed, which motivates the \VDD
approach that mixes two consecutive modes optimally. While computing
optimal speeds is NP-hard with discrete modes, it has polynomial
complexity when mixing speeds. Intuitively, the \VDD approach allows
for smoothing out the discrete nature of the modes.  An alternate (and
simpler in practice) solution to \VDD is the \incremental model, where
one sticks with unique speeds during task execution as in the
\discrete model, but where consecutive modes are regularly
spaced. Such a model can be made arbitrarily efficient, according to
our approximation results.

Altogether, this paper has laid the theoretical foundations for a
comparative study of energy models. In the recent years, we have
observed an increased concern for green computing, and a rapidly
growing number of approaches. It will be very interesting to see which
energy-saving technological solutions will be implemented in
forthcoming future processor chips!

\clearpage
\bibliographystyle{abbrv}
\bibliography{biblio}

\end{document}